\documentclass[submission]{eptcs}
\providecommand{\event}{HCVS 2018}

\usepackage[
  pass,
]{geometry}

\usepackage[square,comma,numbers,sort&compress]{natbib}
\usepackage[usenames,dvipsnames]{xcolor}
\hypersetup{citecolor=blue,linkcolor=blue}
\usepackage{amsthm}
\usepackage{doi}
\usepackage{amsmath,amsopn,amssymb}
\usepackage{subfig}
\usepackage{endnotes,microtype,xspace,graphicx,fancyvrb,multirow}
\usepackage{booktabs}
\usepackage{array,underscore,relsize}
\usepackage[T1]{fontenc}
\usepackage{fancyhdr,lastpage}
\usepackage{enumitem}

\usepackage{aliascnt}
\usepackage{listings}

\usepackage[linesnumbered]{algorithm2e}

\usepackage{semantic}

\usepackage{stmaryrd}

\usepackage{floatrow}

\usepackage{mathtools}
\usepackage{tikz}
\usetikzlibrary{positioning}
\usetikzlibrary{automata}
\usetikzlibrary{shapes}

\newcommand{\sys}{\mbox{\textsc{Shara}}\xspace}

\newcommand{\cc}[1]{\mbox{\smaller[0.5]\texttt{#1}}}

\fvset{fontsize=\scriptsize,xleftmargin=8pt,numbers=left,numbersep=5pt}

\input{code/fmt}

\def\Snospace~{\S{}}

\newif\ifdraft\drafttrue
\newif\ifnotes\notestrue
\ifdraft\else\notesfalse\fi

\input{glyphtounicode}
\newcolumntype{R}[1]{>{\raggedleft\let\newline\\\arraybackslash\hspace{0pt}}p{#1}}

\newcommand{\squishlist}{
\begin{itemize}[noitemsep,nolistsep]
  \setlength{\itemsep}{-0pt}
}
\newcommand{\squishend}{
  \end{itemize}
}

\usepackage{tikz}

\newcommand{\assign}{\mathbin{:=}}

\newcommand{\bigland}{\bigwedge}

\newcommand{\biglor}{\bigvee}

\newcommand{\entails}{\models}

\newcommand{\false}{\mathsf{False}}

\newcommand{\intersection}{\cap}

\newcommand{\none}{\mathsf{None}}

\newcommand{\union}{\cup}

\newtheorem{defn}{\bf{Definition}}

\newtheorem{ex}{\bf{Example}}

\newtheorem{lem}{\bf{Lemma}}

\newtheorem{thm}{\bf{Theorem}}

\newcommand{\bodyof}[1]{\mathsf{Body}(#1)}

\newcommand{\consof}[1]{\mathsf{Constraint}(#1)}

\newcommand{\collapse}[2]{\textsc{Collapse}(#1, #2)}

\newcommand{\copyrel}{\textsc{CopyRel}}

\newcommand{\corr}{\textsc{Corr}}

\newcommand{\ctrof}[1]{\mathsf{Ctr}(#1)}

\newcommand{\depsof}[1]{\mathsf{Deps}(#1)}

\newcommand{\tdepsof}[1]{\mathsf{TrDeps}(#1)}

\newcommand{\duality}{\textsc{Duality}\xspace}

\newcommand{\eldarica}{\textsc{Eldarica}\xspace}

\newcommand{\expand}{\textsc{Expand}}

\newcommand{\expandaux}{\textsc{ExpAux}\xspace}

\newcommand{\expandsto}{\preceq}

\newcommand{\headof}[1]{\mathsf{Head}(#1)}

\newcommand{\impact}{\textsc{Impact}\xspace}

\newcommand{\mcchc}{\cc{S}_{\cc{DA}}}

\newcommand{\nosoln}{\mathsf{None}}

\newcommand{\postctr}{\textsc{Post}}

\newcommand{\prectr}{\textsc{Pre}}

\newcommand{\predof}[1]{\mathsf{Pred}( #1 )}

\newcommand{\seahorn}{\textsc{SeaHorn}\xspace}

\newcommand{\shara}{\textsc{Shara}}

\newcommand{\sharingclause}{\textsc{SharedRel}\xspace}

\newcommand{\solvecdd}{\textsc{SolveCdd}}

\newcommand{\topSort}{\textsc{TopologicalSort}}

\newcommand{\solveitp}{\textsc{Itp}\xspace}

\newcommand{\tformulas}[1]{\mathsf{Forms}( #1 )}

\newcommand{\vc}[1]{\textsc{Cex}(#1)}

\newcommand{\vinta}{\textsc{Vinta}\xspace}

\newcommand{\vocab}{\mathsf{Vocab}}

\newcommand{\whale}{\textsc{Whale}\xspace}

\newcommand{\zthree}{\textsc{Z3}\xspace}

\begin{document}

\title{
  Solving Constrained Horn Clauses Using \\
  Dependence-Disjoint Expansions}

\author{Qi Zhou, David Heath, and William Harris}

\newcommand{\titlerunning}{%
  Solving Constrained Horn Clauses Using
  Dependence-Disjoint Expansions
}

\newcommand{\authorrunning}{%
  Qi Zhou, David Heath, and William Harris
}

\date{}
\maketitle

\begin{abstract}
  Recursion-free Constrained Horn Clauses (CHCs) are logic-programming
  problems that can model safety properties of programs with bounded
  iteration and recursion.
  In addition, many CHC solvers reduce recursive systems to a series
  of recursion-free CHC systems that can each be solved efficiently.

  In this paper, we define a novel class of recursion-free systems,
  named \emph{Clause-Dependence Disjoint} (CDD), that
  generalizes classes defined in previous work.
  The advantage of this class is that many CDD systems are smaller
  than systems which express the same constraints but are part of a
  different class.
  This advantage in size allows CDD systems to be solved more
  efficiently than their counterparts in other classes.
  We implemented a CHC solver named
  \sys.  \sys~solves arbitrary CHC systems by reducing the input to a
  series of CDD systems.
  Our evaluation indicates that \sys~outperforms state-of-the-art
  implementations in many practical cases.
\end{abstract}

\section{Introduction}
\label{sec:intro}
Many critical problems in program verification can be reduced to
solving systems of Constrained Horn Clauses (CHCs), a class of
logic-programming
problems~\cite{bjorner13,flanagan03,rummer13a,rummer13b}.
A CHC is a logical implication with the following form:
$$
  R_1(\vec{v_1}) \leftarrow R_2(\vec{v_2}) \land R_3(\vec{v_3}) \land
  ... \land \varphi(\vec{v_0}, \vec{v_1}, \vec{v_2}, \vec{v_3},...)
$$
Here, the left side of the implication, called the head, contains an
uninterpreted relational predicate applied to a vector of variables.
The right side has any number of such predicates conjoined together
with a \emph{constraint} ($\varphi$). The constraint is a logical formula in a background
theory and may use variables named by the predicates.
A CHC system is a set of CHCs.
The goal of the CHC solving problem is to find suitable interpretations
for each predicate such that each CHC is logically
consistent in isolation.

In this work we focus on the subclass of CHC systems which are known
as \emph{recursion-free}. In a recursion-free CHC system, no
derivation of a predicate will invoke that predicate.
Less formally, a recursion-free CHC system is one where following
implication arrows through the system will never reach the same clause
twice.
Recursion-free CHC systems are an important subclass for two reasons.
First, recursion-free systems can be used to model safety properties
for hierarchical programs~\cite{lal-qadeer15,lal-qadeer-lahiri12}
(programs with only bounded iteration and recursion).
Second and most importantly, a well-known approach for solving a
general CHC system reduces the input problem to solving a sequence
of recursion-free systems.
Such approaches attempt to synthesize a solution for the   original
system from the solutions of recursion-free systems~\cite{bjorner13}.
The performance of such solvers relies
heavily on the performance of solving recursion-free CHC systems.
%

Typically, even recursion-free CHC systems are not solved directly.
Instead, they are reduced to a more specific subclass of
recursion-free CHC system.
These classes include those of
\emph{body-disjoint} (or \emph{derivation tree})
systems~\cite{heizmann10,bjorner13,mcmillan14,rummer13a,rummer13b} and
of \emph{linear} systems~\cite{albarghouthi12a}.
We will discuss these classes in \autoref{sec:overview} and
\autoref{sec:related-work}.
Such classes can be solved by issuing
\emph{interpolation queries} to find suitable definitions for the
uninterpreted predicates.
%

In general, solving a recursion-free CHC system for
propositional logic and the theory of linear integer arithmetic is
co-NEXPTIME-complete~\cite{rummer13b}.
In contrast, solving a linear system or body-disjoint system with the
same logic and theories is in co-NP~\cite{rummer13b}.
We refer to such classes that are solvable in co-NP time as
\emph{directly solvable}.
Because solving an arbitrary recursion-free system is harder than
solving a directly solvable system, solvers which reduce to directly
solvable systems are highly reliant on the size of the reductions.

The first contribution of this paper is the introduction of a novel
class of directly solvable systems that we refer to as
\emph{Clause-Dependence Disjoint} (CDD).
The formal definition of CDD is given at ~\autoref{defn:cdds}.
CDD is a strict superset of the union of previously introduced classes
of directly solvable systems.
The key characteristic of this class is that when an arbitrary
recursion-free system is reduced to a CDD system and to a system from
a different directly solvable class, the CDD system is frequently the
smaller of the two.
Therefore, solving recursion-free systems by reducing them to CDD form
is often less computationally expensive than reducing them to a
system in a different class.

The second contribution of this paper is a solver for CHC systems,
named \sys.
Given a recursion-free system $S$, \sys reduces the problem of solving
$S$ to solving a CDD system $S'$.
In the worst case, it is possible that the size of $S'$ may be
exponential in the size of $S$.
However, empirically we have found that the size of $S'$ is usually
close enough to the size of $S$ that \sys frequently outperforms
\duality, one of the best known CHC solvers.
The procedure implemented in \sys is a generalization of existing
techniques that synthesize compact verification conditions for
hierarchical programs~\cite{flanagan01,lal-qadeer15}.
Given a general (possibly recursive) CHC system, \sys solves a
sequence of recursion-free systems.
Each subsystem is a bounded unwinding of the original system. \sys
attempts to combine the solutions of these recursion-free systems to
synthesize a solution to the original problem, as has been proposed
in previous work~\cite{rummer13b}.

We implemented \sys within the \duality CHC solver~\cite{bjorner13},
which is implemented within the \zthree automatic theorem
prover~\cite{moura08}.
We evaluated the effectiveness of \sys on standard benchmarks drawn
from SVCOMP15~\cite{svcomp15}.
The results indicate that \sys outperforms modern solvers many cases.
Futhermore, the results indicate that combining the strengths of \sys
with that of other existing approaches (as discussed in
\autoref{sec:evaluation}) is a promising direction for the future of
CHC solving.

The rest of this paper is organized as follows.
\autoref{sec:overview} illustrates the operation of \sys on a
recursion-free CHC system.
\autoref{sec:background} reviews technical work on which \sys is
based.
\autoref{sec:approach} describes \sys in technical detail.
\autoref{sec:evaluation} gives the results of our empirical evaluation
of \sys.
\autoref{sec:related-work} compares \sys to related work.

\section{Overview}
\label{sec:overview}

In \autoref{sec:running-ex}, we describe a recursion-free CHC system,
$\mcchc$ (\autoref{fig:chc}), that models the safety of the program
\cc{dblAbs} (\autoref{fig:multicall-code}).
In \autoref{sec:solve-ex}, we show that $\mcchc$ is a CDD system and
how \sys can solve it by encoding it into binary interpolants.
In \autoref{sec:not-in}, we illustrate that $\mcchc$ is not in
directly solvable classes introduced in previous work.
%

\subsection{Verifying \cc{dblAbs}: an example hierarchical program}
\label{sec:running-ex}

\begin{figure}[t]
  \centering
  \begin{floatrow}[2]
    \ffigbox[.3\textwidth] {%
      \caption{\cc{dblAbs}: an example hierarchical program.} %
      \label{fig:multicall-code} }
    {
\begin{lstlisting}[
        basicstyle=\small,
        % xleftmargin=\parindent,
        xleftmargin=2em,framexleftmargin=1.5em,
        numbers=left,
        escapeinside={(*@}{@*)},
        morekeywords={def, return, if, else, assert}]
def dbl(int x)
  return 2 * x
def main(n)
  abs = 0
  if (n >= 0)
     abs = n
  else
     abs = (*@$-$@*)n
  res = dbl(abs)
  assert(res >= 0)
\end{lstlisting}
    }
    \ffigbox[.66\textwidth]{%
      \caption{A CHC system that models the safety condition of $\cc{dblAbs}$,
    named $\mcchc$.}}{%
      \label{fig:chc}
      \begin{align}
        \cc{dbl}(\cc{x}, \cc{d}) &\gets \cc{d} = 2 * \cc{x} \\
        \cc{L}_4(\cc{n}, \cc{abs}) &\gets  \cc{abs} = 0 \\
        \cc{L}_6(\cc{n}, \cc{abs}) &\gets \cc{L}_4(\cc{n}, \cc{abs}) \land \cc{n} \ge 0 \\
        \cc{L}_8(\cc{n}, \cc{abs}) &\gets \cc{L}_4(\cc{n}, \cc{abs}) \land \cc{n} < 0 \\
        \cc{L}_9(\cc{n}, \cc{abs'}) &\gets \cc{L}_6(\cc{n},\cc{abs}) \land \cc{abs'} = \cc{n} \\
        \cc{L}_9(\cc{n}, \cc{abs'}) &\gets \cc{L}_8(\cc{n},\cc{abs}) \land \cc{abs'} = -\cc{n} \\
        \cc{main}(\cc{n},\cc{res}) &\gets \cc{L}_9(\cc{n}, \cc{abs'}) \land
        \cc{dbl}(\cc{x}, \cc{d})
        \land \cc{abs'} = \cc{x} \land \cc{res} = \cc{d} \\
        False &\gets \cc{main}(\cc{n}, \cc{res}) \land \cc{res} < 0 
      \end{align}
    }
  \end{floatrow}
\end{figure}
%
\cc{dblAbs} is a procedure that doubles the absolute value of
its input and stores the result in \cc{res}.
The program also asserts that \cc{res} is greater than or equal to $0$
before exiting.
Verifying this assertion reduces to solving a recursion-free
CHC system over a set of uninterpreted predicates that represent
the control locations in \cc{dblAbs}.
In particular, one such system, $\mcchc$, is shown in
\autoref{fig:chc}.
While $\mcchc$ has been presented as the result
of a translation from \cc{dblAbs}, \sys is
purely a solver for CHC systems: it does not require access to the
concrete representation of a program, or for a given CHC system to be
the result of translation from a program at all.

\subsection{$\mcchc$ as a Clause-Dependence Disjoint System}
\label{sec:solve-ex}
The recursion-free CHC system $\mcchc$ is a
\emph{Clause-Dependence Disjoint} (CDD) system.
A CHC system can be classified as CDD when each clause satisfies the
following rules:
\textbf{(1)} no two predicates in the body of the same clause share
any transitive dependencies on other predicates and
\textbf{(2)} no clause has more than one occurrence of a given
predicate in the body.
As an example, clause $(7)$ is dependence disjoint. 
Two predicates $\cc{L}_9$ and $\cc{dbl}$ are in its body.
The transitive dependency of $\cc{L}_9$ is the set
$\{\cc{L}_4,\cc{L}_6,\cc{L}_8\}$ while the transitive dependency of
\cc{dbl} is the empty set.
Therefore, their transitive dependencies are disjoint:
$\{\cc{L}_4,\cc{L}_6,\cc{L}_8\} \cap \varnothing = \varnothing$.
All other clauses in $\mcchc$ have at most one uninterpreted predicate
in the body, so they are trivially disjoint dependent.
Therefore $\mcchc$ is a CDD system.
The formal definition of CDD and its key properties are given in
\autoref{sec:CDD-defn}.
The formal definition of transitive dependency is given in
\autoref{sec:chcs}.
%

\sys solves CDD systems directly by issuing a binary interpolation
query for each uninterpreted predicate in topological order.
Each interpretation of a predicate $P$ can be computed by
interpolating
\textbf{(1)} the \emph{pre}-formula, constructed from clauses where
$P$ is the head, and
\textbf{(2)} the \emph{post}-formula, constructed from
all clauses where the head transitively depends on $P$.

For example, consider $\cc{L}_9$. By the time \sys attempts to
synthesize an interpretation for $\cc{L}_9$ it will have solutions for
$\cc{L}_4$, $\cc{L}_6$, $\cc{L}_8$.
Possible interpretations of these predicates are shown in
\autoref{fig:ex-graph}.
The pre-formula is constructed from the bodies of clauses where
$\cc{L}_9$ is the head.
Each relational predicate, $P$,
is replaced by a corresponding boolean indicator variable, $\cc{b}_P$.
Each boolean indicator variable implies the solution for its
predicate,
encoded as the disjunction of the negation of the boolean
indicator variable and the solution.
In particular, the pre-formula for $\cc{L}_9$ is constructed from
clauses (5) and (6):
\begin{gather}
  ((\cc{b}_{\cc{L}_6} \land \cc{abs'} = \cc{n})
  \lor
  (\cc{b}_{\cc{L}_8} \land \cc{abs'} = -\cc{n}))
  \land
  (\neg \cc{b}_{\cc{L}_6} \lor \cc{n} \ge 0)
  \land
  (\neg \cc{b}_{\cc{L}_8} \lor \cc{n} < 0)
\end{gather}
The post-formula is constructed from clauses that transitively depend
on $\cc{L}_9$. Again, we replace relational predicates by
corresponding boolean indicators. However, we omit the boolean indicator for
$\cc{L}_9$. The post-formula is composed from clauses (1), (7), and
(8):
\begin{gather}
  (\neg \cc{b}_{\cc{dbl}} \lor \cc{d} = 2*\cc{x})
  \land
  (\neg \cc{b}_{\cc{main}} \lor
    (\cc{b}_{\cc{dbl}}
    \land \cc{abs'}= \cc{x}
    \land \cc{res}=\cc{d} ))
  \land
  (\cc{b}_{\cc{main}} \land \cc{res} < 0)
\end{gather}
Interpolating the pre and post formulas yields an interpretation
of $\cc{L}_9$: $\cc{abs'} \geq 0$.
The procedure for solving a CDD system is described in formal detail
in \autoref{sec:solve-cdd}.

\subsection{$\mcchc$ is not in other recursion-free classes}
\label{sec:not-in}
In this section, we show that $\mcchc$ is not in other
known classes of recursion-free CHC systems. 
Specifically, we will
discuss body-disjoint systems and linear systems.
%

Body-disjoint (or derivation tree)
systems~\cite{mcmillan14,bjorner13,heizmann10,rummer13a,rummer13b} are
a class of recursion-free CHC system where each uninterpreted
predicate appears in the body of at most one clause and appears in
such a clause exactly once.
Such systems cannot model a program with multiple control paths that
share a common prefix, typically modeled as a CHC system with an
uninterpreted predicate that occurs in the body of multiple clauses.
$\mcchc$ is not a body-disjoint system because $\cc{L}_4$ appears in the
body of both clause (3) and clause (4).
In order to handle $\mcchc$, a solver that uses
body-disjoint systems would have to duplicate
$\cc{L}_4$. 
Worse, if $\cc{L}_4$ had dependencies, then each dependency would also
need to be duplicated.

Previous work has also introduced the class of linear
systems~\cite{albarghouthi12a}, where the body of each clause has at
most one uninterpreted predicate.
However, such systems cannot directly model the control flow of a
program that contains procedure calls.
$\mcchc$ is not a linear system because the body of clause (7) has two
predicates, \cc{L9} and \cc{dbl}.
CHC solvers that use linear systems effectively inline the
constraints for relational predicates that occur in non-linear
clauses~\cite{albarghouthi12b}.
In the case of $\mcchc$, inlining the constraints of $\cc{dbl}$ is
efficient, but in general such approaches can generate systems that
are exponentially larger than the input. For example, if the procedure
$\cc{dbl}$ were called more than once in $\cc{dblAbs}$ then multiple
copies of the body of the procedure would be inlined. And if the body
of this procedure were large, the inlining could become
prohibitively expensive.

\begin{figure}[t]
  \centering
  \begin{floatrow}[2]
    \ffigbox[.45\textwidth] {%
      {\caption{$\mcchc$ as a directed hypergraph. Each relational
          predicate is depicted as a graph node while each clause is
          represented by a hyperedge. Each hyperedge is labeled by the
          constraint in the corresponding CHC\@. Each node has a valid
          corresponding interpretation, written in braces.
      }\label{fig:ex-graph}}}
      {\scalebox{0.75}{%
\begin{tikzpicture}
\begin{scope}[every node/.style={circle,thick,draw}]
  \node[label={0:$\{true\}$}] (L4) at (0,0) {$\cc{L}_4$};
  \coordinate[above = of L4] (L4');
  \node[label={180:$\{\cc{n} \geq 0\}$}, below left = of L4] (L6) {$\cc{L}_6$};
  \node[label={0:$\{\cc{n} < 0\}$}, below right = of L4] (L8) {$\cc{L}_8$};
  \node[label={180:$\{\cc{abs'} \geq 0\}$}, below right = of L6] (L9) {$\cc{L}_9$};
  \node[label={0:$\{\cc{res} \geq 0\}$}, below right = 1.6 and 1.2 of L9] (main) {$\cc{main}$};
  \node[label={0:$\{\cc{d} = 2*\cc{x}\}$}, above right = 1.6 and 1.2 of main] (dbl) {$\cc{dbl}$};
  \coordinate[above = of dbl] (dbl');
  \node[label={0:$\{false\}$}, below = of main, style = ultra thick] (bot) {$\bot$};
  \coordinate[above = 0.5 of main] (main');
\end{scope}

\begin{scope}[every edge/.style={draw=black, very thick}]
  \path[->, style=right] (L4') edge node {$\cc{abs} = 0$ \textbf{(2)}} (L4);
  \path[->, style=above left] (L4) edge node {$\cc{n} \geq 0$ \textbf{(3)}} (L6);
  \path[->, style=above right] (L4) edge node {$\cc{n} < 0$ \textbf{(4)}} (L8);
  \path[->, style=left] (L6) edge node {$\cc{abs'} = \cc{n}$ \textbf{(5)}} (L9);
  \path[->, style=right] (L8) edge node {$\cc{abs'} = -\cc{n}$ \textbf{(6)}} (L9);
  \path[-] (L9) edge[out=330, in=90] (main');
  \path[->, style=right] (dbl') edge node {$\cc{d} = 2*\cc{x}$ \textbf{(1)}} (dbl);
  \path[-] (dbl) edge[out=210, in=90] (main');
  \path[->, style=above right] (main') edge node {$\cc{abs'} = \cc{x} \land \cc{res} = \cc{d}$ \textbf{(7)}} (main);
  \path[->, style=right] (main) edge node {$\cc{res} < 0$ \textbf{(8)}} (bot);
\end{scope}
\end{tikzpicture}
}}
    \ffigbox[.45\textwidth]
      {\caption{The hierarchy of classes of recursion-free CHC
          systems. Body Disjoint and Linear systems are subsumed by
          CDD systems. Solving Directly Solvable CHC systems is in
          co-NP while solving general, recursion-free systems is
          co-NEXPTIME Complete.
      }\label{fig:hierarchy}}
      {\begin{tikzpicture}
\begin{scope}
  \node at (0, 0) (rec) {%
    \begin{tabular}{c}
      Recursion Free\\\textsc{co-Nexptime Complete}
    \end{tabular}
  };

  \node[below = 0.5 of rec] (dir) {%
    \begin{tabular}{c}
      Directly Solvable\\\textsc{co-NP}
    \end{tabular}
  };

  \node[draw, rounded corners=0.5cm, below = 0.5 of dir] (cdd) {%
    \begin{tabular}{c}
      Clause-Dependence Disjoint\\\textsc{co-NP}
    \end{tabular}
  };

  \node[below left = 0.5 and -1.1 of cdd] (bd) {%
    \begin{tabular}{c}
      Body-Disjoint\\\textsc{co-NP}
    \end{tabular}
  };

  \node[below right = 0.5 and -1.1 of cdd] (lin) {%
    \begin{tabular}{c}
      Linear\\\textsc{co-NP}
    \end{tabular}
  };

\end{scope}
\begin{scope}[every edge/.style={draw=black, very thick}]
  \path[-] (bd) edge (cdd);
  \path[-] (lin) edge (cdd);
  \path[-] (cdd) edge (dir);
  \path[-] (dir) edge (rec);
\end{scope}
\end{tikzpicture}}
  \end{floatrow}
\end{figure}

The hierarchy of discussed classes of recursion-free systems is
depicted in \autoref{fig:hierarchy}.
As shown, the class of CDD systems is a superset of both the class
of body-disjoint systems and the class of linear systems. So any
recursion-free system that is efficiently expressible in body-disjoint
or linear form is also efficiently expressible in CDD form.
In addition, some systems which are expensive to express in
body-disjoint or linear form are efficiently expressible in CDD form.
\sys~takes advantage of this fact when solving input systems.
Given an arbitrary recursion-free CHC system $S$, \sys~reduces $S$ to
a CDD system $S'$ and solves $S'$ directly.
In general, $S'$ may have size exponential in the size of
$S$.
However, \sys~generates CDD systems via heuristics analogous to
those used to generate compact verification conditions of hierarchical
programs~\cite{flanagan01,lal-qadeer15}.
In practice these heuristics often yield CDD systems which are
small with respect to the input system.
A general procedure for constructing a CDD expansion of a given CHC
system is given in \autoref{app:cons-cdd}.

\section{Background}
\label{sec:background}

\subsection{Constrained Horn Clauses}
\label{sec:chcs}

\subsubsection{Structure}
A Constrained Horn Clause is a logical implication where the
antecedent is called the body and the consequent is called the head.
The body is a conjunction of a logical formula, called the constraint,
and a vector of uninterpreted predicates. The constraint is an arbitrary
formula in some background logic, such as linear integer arithmetic.
The uninterpreted predicates are applied to variables which may or may
not appear in the constraint.
A head can be either an uninterpreted predicate applied to variables or $False$.
A clause where the head is $False$ is called a query. A CHC can be
defined structurally:
\begin{align*}
\cc{chc} \Coloneqq&~\cc{head} \gets \cc{body} \\
\cc{head} \Coloneqq&~False \\
  \mid&~\cc{pred} \\
\cc{body} \Coloneqq&~\varphi \wedge \cc{preds} \\
\cc{preds} \Coloneqq&~True \\
  \mid&~\cc{pred} \wedge \cc{preds} \\
\cc{pred} \Coloneqq&~\textsf{\emph{an uninterpreted predicate applied to variables}} \\
\varphi \Coloneqq&~\textsf{\emph{a formula}} \\
\end{align*}
%
For a given CHC $C$, $\bodyof{C}$ denotes the vector of uninterpreted
predicates in the body and $\consof{C}$ denotes the constraint in the
body.
If $C$ is not a query, then $\headof{C}$ denotes the uninterpreted
predicate in the head.
%
A CHC system is a set of CHCs where exactly one clause is a query.
For a given CHC system $S$, $\predof{S}$ denotes the set of all
uninterpreted predicates and $\cc{query}$ denotes the body of the
query clause.

To explain the structure of a CDD system, we need terminology that relates
predicates in a CHC system including the terms \emph{predicate dependency},
\emph{transitive predicate dependency}, and \emph{sibling}.
\begin{defn}
  Given a CHC system $\cc{S}$ and two uninterpreted predicates
  $\cc{P}$ and $\cc{Q} \in \predof{S}$, if $\exists \cc{C} \in \cc{S}$
  such that $\cc{P} = \headof{C}$ and $\cc{Q} \in \bodyof{C}$, then
  $\cc{Q}$ is a \emph{predicate dependency} of $\cc{P}$.
\end{defn}
\begin{ex}
  In $\mcchc$, because $\cc{L}_4$ is in the body of clause (4) and
  $\cc{L}_8$ is the head of clause (4), $\cc{L}_4$ is a predicate
  dependency of $\cc{L}_8$.
\end{ex}
Given a CHC system $\cc{S}$ and an uninterpreted predicate $\cc{P}$,
$\depsof{P}$ denotes the set of all predicate dependencies of $\cc{P}$
in $\cc{S}$.
%
\begin{defn}
  Given a CHC system $\cc{S}$ and three uninterpreted predicates
  $\cc{P}, \cc{Q}$, and $\cc{R} \in \predof{S}$, if $\cc{Q} \in
  \depsof{P}$ then \cc{Q} is a \emph{transitive predicate dependency} of
  \cc{P}.
  If \cc{Q} is a transitive predicate dependency of \cc{P} and \cc{R}
  is a transitive predicate dependency of \cc{Q}, then $\cc{R}$ is a
  \emph{transitive predicate dependency} of $\cc{P}$.
\end{defn}
\begin{ex}
  In $\mcchc$, because $\cc{L}_4$ is a predicate dependency of
  $\cc{L}_8$, $\cc{L}_4$ is a transitive predicate dependency of
  $\cc{L}_8$.
  And because $\cc{L}_8$ is an transitive predicate dependency of
  $\cc{L}_9$, $\cc{L}_4$ is a transitive predicate dependency of
  $\cc{L}_9$.
\end{ex}
Given a CHC system $\cc{S}$ and an uninterpreted predicate $\cc{P}$,
$\tdepsof{P}$ denotes the set of all transitive predicate dependencies
of $\cc{P}$ in $\cc{S}$.
%
\begin{defn}
  Given a CHC system $\cc{S}$ and two uninterpreted predicates
  $\cc{P}$ and $\cc{Q} \in \predof{S}$, if $\exists \cc{C} \in \cc{S}$
  such that $\cc{P} \in \bodyof{C}$ and $\cc{Q} \in \bodyof{C}$, then
  $\cc{Q}$ and $\cc{P}$ are siblings.
\end{defn}
\begin{ex}
  Because uninterpreted predicates $\cc{L}_9$ and $\cc{dbl}$ both
  appear in the body of clause (7), $\cc{L}_9$ and $\cc{dbl}$ are
  siblings.
\end{ex}
For a given CHC system $\cc{S}$, if there is no uninterpreted
predicate $\cc{P} \in \predof{S}$ such that $\cc{P} \in \tdepsof{P}$,
then $\cc{S}$ is a \emph{recursion-free} CHCs system.

A solution to a CHC system $S$ is a map from each predicate $\cc{P}
\in \predof{S}$ to its corresponding interpretation which is a
formula.
For a solution to be valid, each clause in $S$ must be valid after
substituting each predicate by its interpretation.

\subsection{Logical interpolation}
\label{sec:itps}
All logical objects in this paper are defined over a fixed space of first-order
variables, $X$.
For a theory $T$, the space of $T$
formulas over $X$ is denoted $\tformulas{T}$.
For each formula $\varphi \in \tformulas{T}$, the set of free
variables that occur in $\varphi$ (i.e., the \emph{vocabulary} of
$\varphi$) is denoted $\vocab(\varphi)$.
%
For formulas $\varphi_0, \ldots, \varphi_n, \varphi \in
\tformulas{T}$, the fact that $\varphi_0, \ldots, \varphi_n$
\emph{entail} $\varphi$ is denoted $\varphi_0, \ldots, \varphi_n
\entails \varphi$.

An interpolant of a pair of mutually inconsistent formulas
$\varphi_0$ and $\varphi_1$ in $\tformulas{T}$ is a formula $I$ in $\tformulas{T}$ over 
their common vocabulary that explains their inconsistency.
\begin{defn}
  \label{defn:itps}
  For $\varphi_0, \varphi_1, I \in \tformulas{T}$, if
  \textbf{(1)} $\varphi_0 \entails I$, %
  \textbf{(2)} $I \land \varphi_1 \entails \false$, and %
  \textbf{(3)} $\vocab(I) \subseteq \vocab(\varphi_0) \intersection
  \vocab(\varphi_1)$,
  then $I$ is an \emph{interpolant} of $\varphi_0$ and $\varphi_1$.
\end{defn}
For the remainder of this paper, all spaces of formulas will be
defined for a fixed, arbitrary theory $T$ that supports
interpolation, such as the theory of linear
arithmetic.
Although determining the satisfiability of formulas in such theories
is NP-complete in general, decision procedures~\cite{moura08} and
interpolating theorem provers~\cite{mcmillan04} for such theories have
been proposed that operate on such formulas efficiently.
We define \sys in terms of an abstract interpolating theorem
prover for $T$ named $\solveitp$.
Given two formulas $\varphi_0$ and $\varphi_1$, if $\varphi_0$ and $\varphi_1$ are mutually 
inconsistent, $\solveitp$ returns the interpolant of $\varphi_0$ and $\varphi_1$.
Otherwise, $\solveitp$ returns $\none$.

\section{Technical Approach}
\label{sec:approach}
This section presents the technical details of our approach.
\autoref{sec:CDD-defn} presents the class of Clause-Dependence Disjoint
systems and its key properties.
\autoref{sec:solve-cdd} describes how \sys solves CDD systems
directly.
\autoref{sec:core-solver} describes how \sys solves a given
recursion-free system by solving an CDD system.
Proofs of all theorems stated in this section are in the appendix.

\subsection{Clause-Dependence Disjoint Systems}
\label{sec:CDD-defn}
The key contribution of our work is the introduction of the class of
Clause-Dependence Disjoint (CDD) CHC systems:
\begin{defn}
  \label{defn:cdds}
  For a given recursion-free CHC system $\cc{S}$,
  if for all sibling pairs, $P, Q \in \predof{S}$,
  the transitive dependencies of $P$ and $Q$ are disjoint ($\tdepsof{P} \cap
  \tdepsof{Q} = \emptyset$)
  and no predicate shows more than once in the body of a single clause,
  then $\cc{S}$ is \emph{Clause-Dependence Disjoint (CDD)}.
\end{defn}
CDD systems model hierarchical programs with branches and procedure calls such
that each execution path invokes each statement at most once.
\begin{ex}
  The CHC system $\mcchc$ is a CDD system. An argument is given in
  \autoref{sec:solve-ex}.
\end{ex}

As discussed in \autoref{sec:overview}, CDD is a superset of the union of
the class of body-disjoint systems and the class of linear systems.
For a given recursion-free system $S$, if each uninterpreted predicate $Q \in
\predof{S}$ appears in the body of at most one clause and no
predicate appears more than once in the body of a single clause,
then $S$ is \emph{body-disjoint}~\cite{rummer13a,rummer13b}.
If the body of each clause in $\cc{S}$ contains at most one relational
predicate, then $\cc{S}$ is \emph{linear}~\cite{albarghouthi12a}.

\begin{thm}
\label{thm:cdd-contains}
  The class of CDD systems is a strict superset of the union of the
  class of body-disjoint systems and the class of linear systems.
\end{thm}
Proof is given in \autoref{app:char}.

\subsection{Solving a CDD system}
\label{sec:solve-cdd}

\begin{algorithm}[t]
  \SetKwInOut{Input}{Input}
  \SetKwInOut{Output}{Output}
  \SetKwProg{myproc}{Procedure}{}{}
  \Input{A CDD System $\cc{S}$.}
  \Output{%
    If $\cc{S}$ is solvable, then a solution of
    $\cc{S}$;
    otherwise, the value $\nosoln$.
  }
  \myproc{$\solvecdd(\cc{S})$}{%
    $\sigma \assign \emptyset$ \\
    $\cc{Preds} \assign \topSort(\predof{S})$\label{line:top-sort} \\
    \For {$P \in \cc{Preds}$}{%
      $\cc{interpolant} \assign
        \solveitp(\prectr(P,\sigma),\postctr(P,\sigma))$\label{line:interp} \\
      \Switch{\cc{interpolant}}{%
        \lCase{\cc{SAT}:}{%
          \Return{$\nosoln$}
        }\label{line:interp-sat}
        \lCase{$\cc{I}$:}{%
          $\sigma$ [$P$] $\assign$ $\cc{I}$
        }\label{line:interp-valid}
      }
    }
    \Return{$\sigma$}\label{line:solve-done}
  }
  \caption{$\solvecdd$: for a CDD system $\cc{S}$, returns a
    solution to $\cc{S}$ or the value $\none$ to denote that
    $\cc{S}$ has no solution.}
  \label{alg:solve-cdd}
\end{algorithm}

\autoref{alg:solve-cdd} presents $\solvecdd$, a procedure designed to solve
CDD systems.
Given a CDD system $\cc{S}$, $\solvecdd$ topologically sorts the uninterpreted
predicates in $\cc{S}$ based on their dependency relations
(\autoref{line:top-sort}).
Then, the algorithm calculates interpretations for each predicate in this order
by invoking $\solveitp$ (\autoref{line:interp}).
$\solveitp$ computes a binary interpolant of the pre and post formulas of the
given predicate, where these formulas are based on the current, partial
solution.
The pre and post formulas are computed respectively by $\prectr$ and
$\postctr$, which we define in \autoref{sec:cons-pre} and \autoref{sec:cons-post}.
It is possible that the pre and post formulas may be mutually
satisifiable, in which case $\solveitp$ returns \cc{SAT}
(\autoref{line:interp-sat}). In this case, $\solvecdd$ returns
$\nosoln$ to indicate that $\cc{S}$ is not solvable.
Otherwise, $\solvecdd$ updates the partial solution by setting the
interpretation of $P$ to $\cc{I}$ (\autoref{line:interp-valid}).
Once all predicates have been interpolated, $\solvecdd$ returns the complete
solution, $\sigma$ (\autoref{line:solve-done}).

\begin{ex}
  Given the CDD system $\mcchc$, $\solvecdd$ may generate interpolation
  queries in any topological ordering of the dependency relations.
  One such ordering is $\cc{L}_4$, $\cc{L}_6$, $\cc{L}_8$, $\cc{L}_9$, \cc{dbl},
  \cc{main}.
\end{ex}

\begin{thm}
  \label{thm:solve-cp}
  Given a CDD system $S$ over the theory of linear integer arithmetic,
  $\solvecdd$ either returns the solution of $S$ or $\none$ in co-NP
  time.
\end{thm}
Proof is given in \autoref{app:solve-cp}.

In order to solve a CDD system, we construct efficiently sized pre and
post formulas for each relational predicate and interpolate over these
formulas. These pre and post formulas are built from \textbf{(1)} the
\emph{constraint} of a given predicate, which explains under what
conditions the predicate holds, and \textbf{(2)} the
\emph{counterexample characterization}, which explains what condition
must be true if the predicate holds.

\subsubsection{Constructing constraints for predicates}
In order to construct efficiently sized pre and post formulas for
relational predicates, we use a method for compactly expressing the
constraints on a given predicate.
For a CDD system $S$, a predicate $P \in \predof{S}$, and a partial
solution $\sigma$ that maps predicates to their solutions, the formula
$\ctrof{P,\sigma}$ is a compact representation of the constraints of
$P$.
If $\sigma$ does not contain $P$, then the constraint of $P$ is
constructed from the clauses where $P$ is the head.  When $\sigma$
does contain $P$, $\ctrof{P,\sigma}$ is a lookup from $\sigma$.
Each $P \in \predof{S}$ has a corresponding boolean variable
$\cc{b}_P$:
\[
  \ctrof{P,\sigma}=
  \begin{dcases}
    \biglor_{(\cc{C}_i \in \cc{S}) \land (\headof{C_i} = P)}
    \left( \consof{C_i} \land %
      \bigland_{ \cc{Q} \in \bodyof{C_i}} \cc{b}_Q
    \right),
  &\text{if } P \notin \sigma\\
  \sigma[P], &\text{if } P \in \sigma
  \end{dcases}
\]
%
The \emph{counterexample characterization} of $P$ is a small extension
of the compact constraint of $P$.
It states that if $P$ is used
(meaning $\cc{b}_{P} = True$), then the constraint of $P$ must hold:
\[
  \vc{P,\sigma} = \neg \cc{b}_{P} \lor \ctrof{P,\sigma}
\]

\begin{ex}
  \label{ex:ctr}
  When $\solvecdd$ solves predicate $\cc{L}_9$ in $\mcchc$, it generates a
  constraint based on clauses (5) and (6):
  $$\ctrof{\cc{L}_9,\sigma} =
    (\cc{abs'} = \cc{n} \land \cc{b}_{\cc{L}_6})
    \lor
    (\cc{abs'} = -\cc{n} \land \cc{b}_{\cc{L}_8})$$
  The counterexample characterization for $\cc{L}_9$ is based on its boolean
  indicator and its constraint:
  $$\vc{\cc{L}_9,\sigma} = \neg \cc{b}_{\cc{L}_9} \lor \ctrof{\cc{L}_9,\sigma}$$
\end{ex}

\subsubsection{Constructing pre-formulas for predicates}
\label{sec:cons-pre}
$\prectr(P, \sigma)$ denotes the pre-formula for an arbitrary predicate $P$
with respect to the partial solution map, $\sigma$.
Due to the topological ordering, when $\solvecdd$ attempts to solve $P$, the
interpretations for all dependencies of $P$ will be stored in
$\sigma$.
The pre-formula is built from these interpretations together with boolean
indicators of the dependencies and the constraint of $P$:
\[
\prectr(P,\sigma) =
  \ctrof{P,\sigma} \land
    \left(
      \bigland_{\cc{Q} \in \depsof{P}}
      \left(\neg b_Q \lor \sigma[ \cc{Q} ]  \right)
    \right)
\]
\begin{ex}
  \label{ex:pre-ctr}
  When $\solvecdd$ solves predicate $\cc{L}_9$ in $\mcchc$,
  $\sigma$ maps $\cc{L}_6$ to $n \ge 0$ and $\cc{L}_8$ to $n < 0$.
  The pre-formula for $\cc{L}_9$ under $\sigma$ is therefore:
  $$\ctrof{\cc{L}_9}
  \land (\neg \cc{b}_{\cc{L}_6} \lor n \ge 0) \land (\neg \cc{b}_{\cc{L}_8}
  \lor n < 0)$$
  The formula $\ctrof{\cc{L}_9,\sigma}$ is given in~\autoref{ex:ctr}.
\end{ex}

\subsubsection{Constructing post-formulas for predicates}
\label{sec:cons-post}
$\postctr(P, \sigma)$ denotes the post-formula for an arbitrary predicate $P$
with respect to the partial solution map, $\sigma$. A valid
post-formula is mutually inconsistent with the solution of $P$.
and is constructed based on the predicates which depend on $P$.
%
Let $D_0$ be the \emph{transitive dependents} of $P$ in $S$ (i.e the
predicates that have $P$ as a transitive dependency),
let $D_1$ be the siblings in $S$ of $(D_0 \union {P})$,
let $D_2$ be all transitive dependencies of $D_1$, and
let $D = D_0 \union D_1 \union D_2$.
The post-formula for $P$ under $\sigma$ is the conjunction of counterexample
characterization of all predicates $Q \in D$ and the query clause:
\[
\postctr(R, \sigma) =\cc{query} \land (\bigland_{\cc{Q} \in D} \vc{Q,\sigma})
\]
\begin{ex}
  \label{ex:ctx-ctr}
  When $\solvecdd$ solves $\cc{L}_9$ in $\mcchc$, it must consider
  the dependents of $\cc{L}_9$.
  The \emph{transitive dependents} $D_0$ of $\cc{L}_9$ is \{\cc{main}\}.
  The siblings set $D_1$ is \{\cc{dbl}\}.
  The set of transitive dependencies of $D_1$ is $\varnothing$. 
  Therefore, $D$ is \{\cc{main},\cc{dbl}\}.
  The \cc{query} is $\cc{b}_{main} \land res<0$.
  The post-formula for $\cc{L}_9$ under $\sigma$ is:
  $$\cc{query} \land \vc{\cc{main},\sigma} \land \vc{\cc{dbl},\sigma}$$
\end{ex}

\subsection{Solving recursion-free systems using CDD systems}
\label{sec:core-solver}

Given a recursion-free CHC system $S$, \shara~constructs a CDD system
$S'$.
\sys then directly solves $S'$ and, from this solution, constructs a solution
for $S$.
For two given recursion-free CHC systems $S$ and $S'$, if
there is a homomorphism from $\predof{S'}$ to $\predof{S}$ that
preserves the relationship between the clauses of $S'$ in
the clauses of $S$, then $S'$ is an \emph{expansion} of
$S$ (all definitions in this section will be over fixed
$S$, and $S'$).
\begin{defn}
  \label{defn:expansion}
  Let $\eta : \predof{S'} \to \predof{S}$ be such that
  \textbf{(1)} for all $P' \in \predof{S'}$, $P'$ has the
  same parameters as $\eta(P')$;
  \textbf{(2)} for each clause $C' \in S'$, the clause $C$, constructed
  by substituting all predicates $P'$ by $\eta(P')$, is in $S$; and
  \textbf{(3)} each predicate $P$ in $S$ has at least one predicate
  $P'$ in $S'$ such that $\eta(P') = P$.
  Then $\eta$ is a \emph{correspondence} from $S'$ to $S$.
\end{defn}
If there is a correspondence from $S'$ to $S$, then $S'$ is an \emph{expansion}
of $S$, denoted $S \expandsto S'$.

\begin{defn}
  \label{defn:min-expansion}
  If $S'$ is CDD, $S \expandsto S'$, and there is no CDD system $S''$ such that $S \expandsto S' \expandsto S''$ and
  $S'' \neq S'$, then $S'$ is a minimal CDD expansion of $S$.
\end{defn}

\begin{algorithm}[t]
  \SetKwInOut{Input}{Input}
  \SetKwInOut{Output}{Output}
  \SetKwProg{myproc}{Procedure}{}{}
  \Input{A recursion-free CHC system $S$.}
  \Output{A solution to $S$ or $\none$.}
  \myproc{$\sys(S)$ %
    \label{line:shara-begin}}{ %
    $(S', \eta) \assign %
    \expand(S)$ \label{line:shara-expand} \; %
    \Switch{ $\solvecdd(S' )$ %
      \label{line:shara-case} }{
      \lCase{ $\none$: }{ \Return{$\none$} \label{line:shara-ret-none} } %
      \lCase{ $\sigma'$: }{ %
        \Return{ $\collapse{ \eta }{ \sigma' }$ } %
        \label{line:shara-ret-collapse}
      } %
    } %
  } %
  \caption{\sys: a solver for recursion-free CHCs, which uses
    procedures $\expand{}$ (see \autoref{app:cons-cdd}) and
    $\solvecdd$ (see
    \autoref{sec:solve-cdd}). }
  \label{alg:shara}
\end{algorithm}
\sys (\autoref{alg:shara}), given a recursion-free CHC system $S$
(\autoref{line:shara-begin}), 
returns a solution
to $S$ or the value $\none$ to denote that $S$ is
unsolvable.
\sys first runs a procedure $\expand$ on $S$ to obtain a CDD expansion
$S'$ of $S$ ($\expand$ is given in \autoref{app:cons-cdd}).
\sys then invokes $\solvecdd$ on $S'$.
When \solvecdd~returns that $S'$ has
no solution, \sys~propagates $\none$ (\autoref{line:shara-ret-none}).
Otherwise, \sys constructs a solution from the CDD solution,
$\sigma'$, by invoking $\collapse{ \eta }{ \sigma' }$
(\autoref{line:shara-ret-collapse}).
$\textsc{Collapse}$ is designed to convert the solution for the CDD system
back to a solution for the original problem. It does this by
taking the conjunction of all interpretations of predicates which
correspond to the same predicate in the original problem.
That is, given a CDD solution $\sigma'$ and a correspondence $\eta$ from
$P' \in \predof{S'}$ to $P \in \predof{S}$,
$\collapse{\eta}{\sigma'}$ generates an entry in $\sigma$ for each
predicate in the original system:
$\sigma[P] \assign \bigland_{ \eta(P') = P} \sigma'(P')$.

\begin{thm}
  \label{thm:corr}
  $S$ is solvable if and only if $\shara$ returns a solution $\sigma$.
\end{thm}
Proof is given in \autoref{app:corr}.


\section{Evaluation}
\label{sec:evaluation}
We performed an empirical evaluation of \sys to determine how it
compares to existing CHC solvers.
To do so, we implemented \sys as a modification of \duality CHC
solver, which is included in the \zthree theorem prover~\cite{z3}.
We modified \duality to use \sys as its solver for recursion-free CHC
systems.
We modified the algorithm used by \duality to generate recursion-free
unwindings of a given recursive system so that, in each iteration, it
generates an unwinding which is converted to CDD form.
In the following context, ``\sys'' refers to this modified version of
\duality.

We evaluated \sys and an unmodified version of \duality on 4,309 CHC
systems generated from programs in the SV-COMP 2015~\cite{svcomp15}
verification benchmark suite.
To generate CHC systems, we ran the \seahorn~\cite{gurfinkel15}
verification framework with its default settings (procedures are not
inlined and each loop-free fragment is a clause), set to timeout at 90
seconds.
We used the benchmarks in SV-COMP 2015~\cite{svcomp15} because they
were used to evaluate \duality in previous work~\cite{mcmillan14}.

We also compared \sys to the \eldarica CHC solver, but
\eldarica could not parse the CHC systems generated by \seahorn.
We compared \sys and \eldarica on an alternative set of benchmarks
generated by the UFO model checker~\cite{albarghouthi12c}, and found
that \sys outperformed \eldarica by at least an order of magnitude on
an overwhelming number of cases.
As a result, we focus our discussion on a comparison of \sys and
\duality.

All experiments were run using a single thread 
on a machine with 16 1.4 GHz processors and
128 GB of RAM.
We ran the solvers on each benchmark, timing out each implementation
after 180 seconds.

Out of 4,309 benchmarks, \sys solved or refuted 2,408 while \duality
solved or refuted 2,321. \sys timed out on 762 benchmarks and \duality
timed out on 1,145.
On the remaining benchmarks, some constraint caused \zthree's
interpolating theorem prover to fail, meaning the result was neither a
solve nor a refutation. \sys reached this failure on 1,139 benchmarks
while \duality failed on 843.
The two solvers can induce a failure in \zthree on different systems
because in attempting to solve a given system, they generate different
interpolation queries.

\begin{figure}[t]
  \centering
  \begin{floatrow}[2]
    \ffigbox[.48\textwidth] %
    {\caption{Solving times of \sys vs. \duality.
        The $x$ and $y$ axes range over the solving times in seconds
        of \sys and \duality, respectively.
        Each point depicts the performance of a benchmark.
        The line $y = x$ is shown in red.} %
      \label{fig:complete-data} }
    { \includegraphics[width=\linewidth]{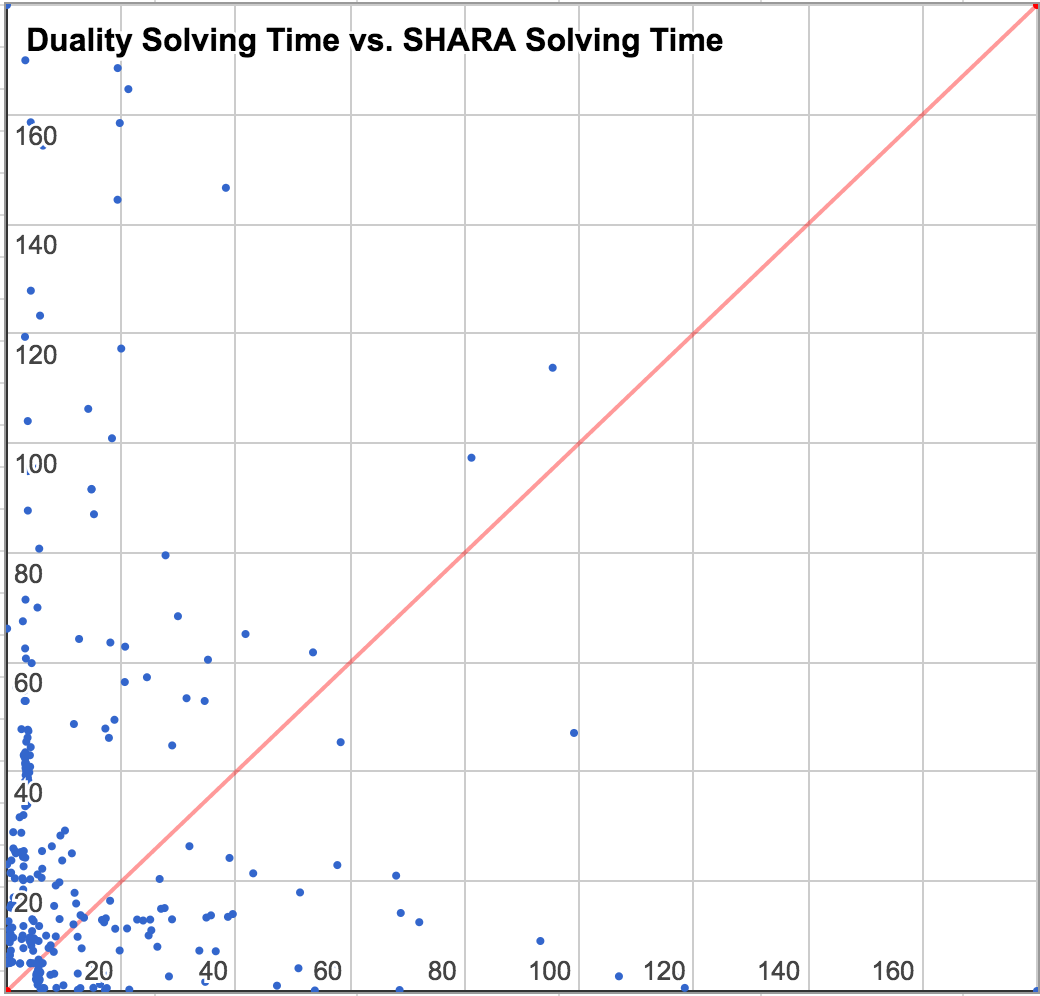} }
    \ffigbox[.48\textwidth] %
    { \caption{Times of \sys and \duality vs. system size.
        The $x$-axis ranges over the size of a given system, and the
        $y$-axis ranges over solvers' times.
        Measurements of \sys and \duality are shown in blue and
        red, respectively. } %
      \label{fig:size} }
     { \includegraphics[width=\linewidth]{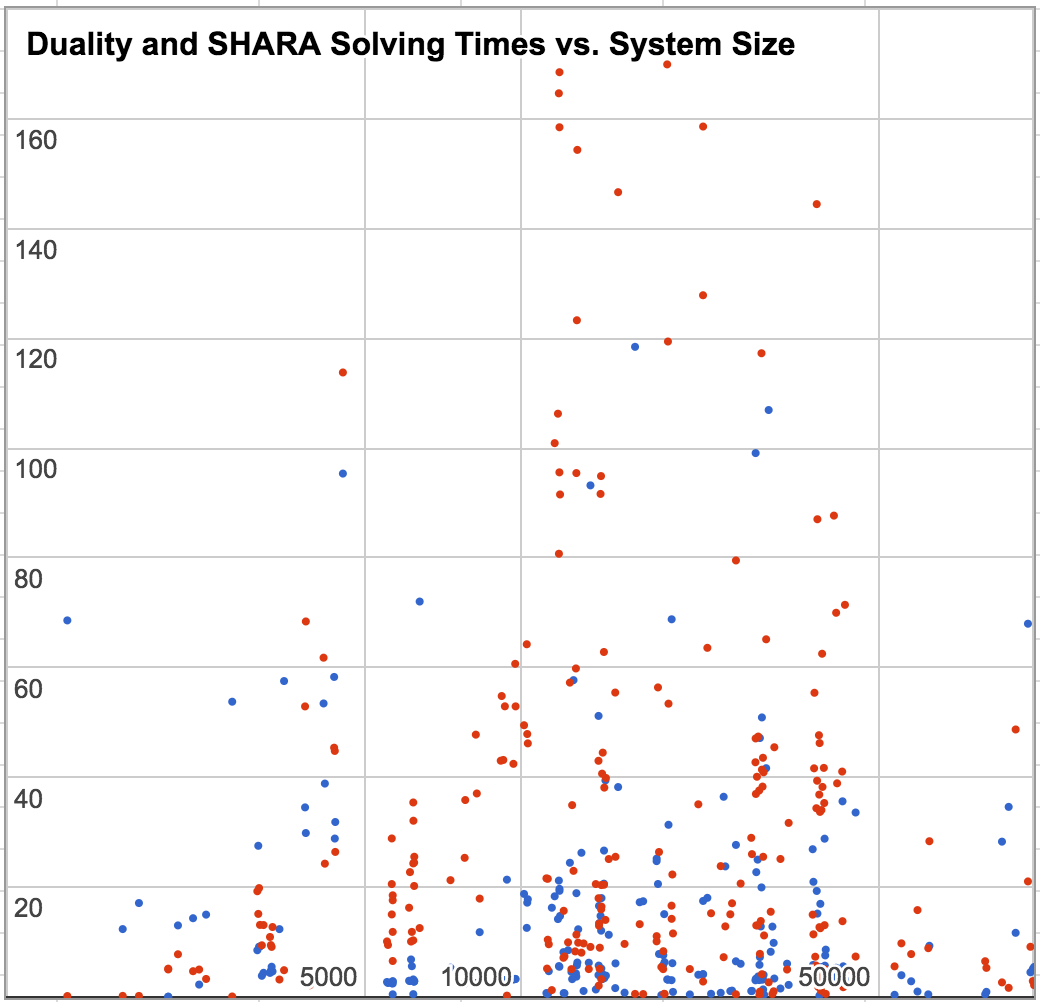} }
  \end{floatrow}
\end{figure}
The results of our evaluation are shown in
\autoref{fig:complete-data}.
Of the 4,040 benchmarks on which both solvers took a short amount
of time---less than five seconds---\sys solved the benchmarks in an
average of 0.51 seconds and \duality solved them in an average of 0.42
seconds.
\autoref{fig:complete-data} contains data for benchmarks which took
longer than five second for both systems to solve.
Out of these 269 benchmarks, \sys solved 185 in less time
than \duality, and solved 159 in less than half the time of \duality.
\duality solved 84 in less time than \sys, and solved 53 in less than
half the time of \sys.
Of the 762 benchmarks on which \sys timed out, \duality solved or
found a counterexample to 185.
Of the 1,145 benchmarks on which \duality timed out, \sys solved or
found a counterexample to 470.

\autoref{fig:size} shows the relationship between the solving times of
\duality and \sys and the size of a given system, measured as lines of
code in the format generated by \seahorn.
The majority of files have between 1,000 and 100,000 lines, so
\autoref{fig:size} is restricted to this range.
The data indicates that the performance improvement of \sys compared
to \duality is consistent across systems of all sizes available.

The results indicate that \textbf{(1)} on a significant number of
different verification problems, \sys can perform significantly better than
\duality, but
\textbf{(2)} there are some cases in which the strengths of each
algorithm yield better results.
Specifically, the unmodified version of \duality~uses a technique
called lazy annotation to avoid enumerating all derivation trees.
\sys~does not use lazy annotation.
We collected the differences between sizes of a given system and its
minimal CDD expansion generated by \sys and found that they were
independent of \sys's performance compared to \duality.
Thus, while \duality may in the worst case enumerate exponentially
many derivation trees, it appears to enumerate far fewer
than the worst-case bound in some cases, causing it to perform better
than \sys.
Our results indicate that a third approach that combines the strengths
of both \duality and \sys, perhaps by \emph{lazily} unwinding a given
system into a series of CDD systems, could yield further improvements.


\section{Related Work}
\label{sec:related-work}
A significant body of previous work has presented solvers for
different classes of Constrained Horn Clauses, or finding inductive
invariants of programs that correspond to solutions of CHCs.
\impact attempts to verify a given sequential procedure by iteratively
selecting paths and synthesizing invariants for each path. This
approach corresponds to solving a recursive linear CHC
system~\cite{mcmillan06}.

Previous work also proposed a verifier for recursive
programs~\cite{heizmann10}.
The proposed approach selects interprocedural paths of a program and
synthesizes invariants for each as nested interpolants.
Such an approach corresponds to attempting to solve a recursive CHC
system $S$ by selecting derivation trees of $S$
and solving each tree.

Previous work has proposed solvers for recursive systems that, given a
system $S$, attempt to solve $S$ by generating and
solving a series of recursion-free unwindings of $S$.
In particular, \eldarica attempts to solve each unwinding
$S'$ by reducing to and solving body-disjoint systems~\cite{rummer13a,rummer13b}.
\duality attempts to avoid solving all derivation-trees (i.e
body-disjoint systems) by using lazy annotation~\cite{bjorner13}.
Other optimizations select derivation trees to solve using
symbolic analogs of Prolog evaluation with
tabling~\cite{jaffar09,mcmillan14}.

\whale attempts to verify sequential recursive programs by generating
and solving hierarchical programs, which correspond to recursion-free CHC
systems~\cite{albarghouthi12b}.
To solve a particular recursion-free system, \whale solves a linear
inlining of the input using a procedure named
\vinta~\cite{albarghouthi12a}.
In general, the linear inlining may be exponentially larger than the
input.

\sys is similar to the recursion-free CHC approaches given above in
that it reduces the problem to solving a CHC system in a
directly-solvable class.
\sys is distinct in that it reduces to solving Clause-Dependent
Disjoint (CDD) systems.
As discussed in \autoref{sec:overview}, the class of CDD systems is a
superset of classes used by the approaches above. CDD systems can
also be solved directly.

\sys solves general CHC systems using the same strategy as proposed by
the above approaches.
Specifically, it solves a series of recursion-free unwindings of the
original system, and tries to synthesize a general solution from the
recursion-free solutions.

Previous work describes solvers for non-linear Horn clauses over
particular theories.
In particular, verifiers have been proposed for recursion-free systems
over the theory of linear arithmetic~\cite{komuravelli14}.
Because the verifier relies on quantifier elimination, it is not clear
if it can be extended to richer theories that support interpolation,
such as the combination of linear arithmetic with uninterpreted
functions.
Other work describes a solver for the class of \emph{timed pushdown
systems}, a subclass of CHC systems over the theory of linear real
arithmetic~\cite{hoder12}.
Unlike these approaches, \sys can solve systems over any theory that
supports interpolation.

DAG inlining attempts to generate compact verification conditions for
hierarchical programs~\cite{lal-qadeer15}.
\sys attempts to solve recursion-free CHC systems by reducing them to
compact CDD systems.
Because hierarchical programs and recursion-free CHC systems are
closely related, algorithms that operate on hierarchical programs
correspond to algorithms that operate on recursion-free Horn Clauses.
However, it is not apparent whether such algorithms can be used directly
to synthesize solutions.

\bibliographystyle{eptcs}
\small

\bibliography{p}

\begin{thebibliography}{10}
\providecommand{\bibitemdeclare}[2]{}
\providecommand{\surnamestart}{}
\providecommand{\surnameend}{}
\providecommand{\urlprefix}{Available at }
\providecommand{\url}[1]{\texttt{#1}}
\providecommand{\href}[2]{\texttt{#2}}
\providecommand{\urlalt}[2]{\href{#1}{#2}}
\providecommand{\doi}[1]{doi:\urlalt{http://dx.doi.org/#1}{#1}}
\providecommand{\bibinfo}[2]{#2}

\bibitemdeclare{inproceedings}{albarghouthi12a}
\bibitem{albarghouthi12a}
\bibinfo{author}{Aws \surnamestart Albarghouthi\surnameend},
  \bibinfo{author}{Arie \surnamestart Gurfinkel\surnameend} \&
  \bibinfo{author}{Marsha \surnamestart Chechik\surnameend}
  (\bibinfo{year}{2012}): \emph{\bibinfo{title}{Craig Interpretation}}.
\newblock In: {\sl \bibinfo{booktitle}{{SAS}}},
  \doi{10.1007/978-3-642-33125-1_21}.

\bibitemdeclare{inproceedings}{albarghouthi12b}
\bibitem{albarghouthi12b}
\bibinfo{author}{Aws \surnamestart Albarghouthi\surnameend},
  \bibinfo{author}{Arie \surnamestart Gurfinkel\surnameend} \&
  \bibinfo{author}{Marsha \surnamestart Chechik\surnameend}
  (\bibinfo{year}{2012}): \emph{\bibinfo{title}{Whale: An Interpolation-Based
  Algorithm for Inter-procedural Verification}}.
\newblock In: {\sl \bibinfo{booktitle}{{VMCAI}}},
  \doi{10.1007/978-3-642-27940-9_4}.

\bibitemdeclare{inproceedings}{albarghouthi12c}
\bibitem{albarghouthi12c}
\bibinfo{author}{Aws \surnamestart Albarghouthi\surnameend},
  \bibinfo{author}{Yi~\surnamestart Li\surnameend}, \bibinfo{author}{Arie
  \surnamestart Gurfinkel\surnameend} \& \bibinfo{author}{Marsha \surnamestart
  Chechik\surnameend} (\bibinfo{year}{2012}): \emph{\bibinfo{title}{{UFO}: {A}
  Framework for Abstraction- and Interpolation-Based Software Verification}}.
\newblock In: {\sl \bibinfo{booktitle}{{CAV}}},
  \doi{10.1007/978-3-642-31424-7_48}.

\bibitemdeclare{inproceedings}{bjorner13}
\bibitem{bjorner13}
\bibinfo{author}{Nikolaj \surnamestart Bj{\o}rner\surnameend},
  \bibinfo{author}{Kenneth~L. \surnamestart McMillan\surnameend} \&
  \bibinfo{author}{Andrey \surnamestart Rybalchenko\surnameend}
  (\bibinfo{year}{2013}): \emph{\bibinfo{title}{On Solving Universally
  Quantified {Horn} Clauses}}.
\newblock In: {\sl \bibinfo{booktitle}{{SAS}}},
  \doi{10.1007/978-3-642-38856-9_8}.

\bibitemdeclare{inproceedings}{flanagan03}
\bibitem{flanagan03}
\bibinfo{author}{Cormac \surnamestart Flanagan\surnameend}
  (\bibinfo{year}{2003}): \emph{\bibinfo{title}{Automatic Software Model
  Checking Using {CLP}}}.
\newblock In: {\sl \bibinfo{booktitle}{{ESOP}}},
  \doi{10.1007/3-540-36575-3_14}.

\bibitemdeclare{inproceedings}{flanagan01}
\bibitem{flanagan01}
\bibinfo{author}{Cormac \surnamestart Flanagan\surnameend} \&
  \bibinfo{author}{James~B. \surnamestart Saxe\surnameend}
  (\bibinfo{year}{2001}): \emph{\bibinfo{title}{Avoiding exponential explosion:
  generating compact verification conditions}}.
\newblock In: {\sl \bibinfo{booktitle}{{POPL}}}, \doi{10.1145/360204.360220}.

\bibitemdeclare{misc}{svcomp15}
\bibitem{svcomp15}
 (\bibinfo{year}{2016}): \emph{\bibinfo{title}{GitHub - sosy-lab/sv-benchmarks:
  svcomp15}}.
\newblock
  \bibinfo{howpublished}{\url{https://github.com/sosy-lab/sv-benchmarks/releases/tag/svcomp15}}.
\newblock \bibinfo{note}{Accessed: 2016 Dec 13}.

\bibitemdeclare{misc}{z3}
\bibitem{z3}
 (\bibinfo{year}{2016}): \emph{\bibinfo{title}{GitHub - Z3Prover/z3: The Z3
  Theorem Prover}}.
\newblock \bibinfo{howpublished}{\url{https://github.com/Z3Prover/z3}}.
\newblock \bibinfo{note}{Accessed: 2016 Dec 13}.

\bibitemdeclare{inproceedings}{gurfinkel15}
\bibitem{gurfinkel15}
\bibinfo{author}{Arie \surnamestart Gurfinkel\surnameend},
  \bibinfo{author}{Temesghen \surnamestart Kahsai\surnameend},
  \bibinfo{author}{Anvesh \surnamestart Komuravelli\surnameend} \&
  \bibinfo{author}{Jorge~A. \surnamestart Navas\surnameend}
  (\bibinfo{year}{2015}): \emph{\bibinfo{title}{The {SeaHorn} Verification
  Framework}}.
\newblock In: {\sl \bibinfo{booktitle}{{CAV}}},
  \doi{10.1007/978-3-319-21690-4_20}.

\bibitemdeclare{inproceedings}{heizmann10}
\bibitem{heizmann10}
\bibinfo{author}{Matthias \surnamestart Heizmann\surnameend},
  \bibinfo{author}{Jochen \surnamestart Hoenicke\surnameend} \&
  \bibinfo{author}{Andreas \surnamestart Podelski\surnameend}
  (\bibinfo{year}{2010}): \emph{\bibinfo{title}{Nested interpolants}}.
\newblock In: {\sl \bibinfo{booktitle}{{POPL}}}, \doi{10.1145/1706299.1706353}.

\bibitemdeclare{inproceedings}{hoder12}
\bibitem{hoder12}
\bibinfo{author}{Krystof \surnamestart Hoder\surnameend} \&
  \bibinfo{author}{Nikolaj \surnamestart Bj{\o}rner\surnameend}
  (\bibinfo{year}{2012}): \emph{\bibinfo{title}{Generalized Property Directed
  Reachability}}.
\newblock In: {\sl \bibinfo{booktitle}{{SAT}}},
  \doi{10.1007/978-3-642-31612-8_13}.

\bibitemdeclare{inproceedings}{jaffar09}
\bibitem{jaffar09}
\bibinfo{author}{Joxan \surnamestart Jaffar\surnameend},
  \bibinfo{author}{Andrew~E. \surnamestart Santosa\surnameend} \&
  \bibinfo{author}{Razvan \surnamestart Voicu\surnameend}
  (\bibinfo{year}{2009}): \emph{\bibinfo{title}{An Interpolation Method for
  {CLP} Traversal}}.
\newblock In: {\sl \bibinfo{booktitle}{{CP}}},
  \doi{10.1007/978-3-642-04244-7_37}.

\bibitemdeclare{inproceedings}{komuravelli14}
\bibitem{komuravelli14}
\bibinfo{author}{Anvesh \surnamestart Komuravelli\surnameend},
  \bibinfo{author}{Arie \surnamestart Gurfinkel\surnameend} \&
  \bibinfo{author}{Sagar \surnamestart Chaki\surnameend}
  (\bibinfo{year}{2014}): \emph{\bibinfo{title}{{SMT}-Based Model Checking for
  Recursive Programs}}.
\newblock In: {\sl \bibinfo{booktitle}{{CAV}}},
  \doi{10.1007/978-3-319-08867-9_2}.

\bibitemdeclare{inproceedings}{lal-qadeer15}
\bibitem{lal-qadeer15}
\bibinfo{author}{Akash \surnamestart Lal\surnameend} \& \bibinfo{author}{Shaz
  \surnamestart Qadeer\surnameend} (\bibinfo{year}{2015}):
  \emph{\bibinfo{title}{{DAG} inlining: a decision procedure for
  reachability-modulo-theories in hierarchical programs}}.
\newblock In: {\sl \bibinfo{booktitle}{{PLDI}}}, \doi{10.1145/2813885.2737987}.

\bibitemdeclare{inproceedings}{lal-qadeer-lahiri12}
\bibitem{lal-qadeer-lahiri12}
\bibinfo{author}{Akash \surnamestart Lal\surnameend}, \bibinfo{author}{Shaz
  \surnamestart Qadeer\surnameend} \& \bibinfo{author}{Shuvendu~K.
  \surnamestart Lahiri\surnameend} (\bibinfo{year}{2012}):
  \emph{\bibinfo{title}{A Solver for Reachability Modulo Theories}}.
\newblock In: {\sl \bibinfo{booktitle}{{CAV}}},
  \doi{10.1007/978-3-642-31424-7_32}.

\bibitemdeclare{inproceedings}{mcmillan04}
\bibitem{mcmillan04}
\bibinfo{author}{Kenneth~L. \surnamestart McMillan\surnameend}
  (\bibinfo{year}{2004}): \emph{\bibinfo{title}{An Interpolating Theorem
  Prover}}.
\newblock In: {\sl \bibinfo{booktitle}{{TACAS}}},
  \doi{10.1016/j.tcs.2005.07.003}.

\bibitemdeclare{inproceedings}{mcmillan06}
\bibitem{mcmillan06}
\bibinfo{author}{Kenneth~L. \surnamestart McMillan\surnameend}
  (\bibinfo{year}{2006}): \emph{\bibinfo{title}{Lazy Abstraction with
  Interpolants}}.
\newblock In: {\sl \bibinfo{booktitle}{CAV}}, \doi{10.1007/11817963_14}.

\bibitemdeclare{inproceedings}{mcmillan14}
\bibitem{mcmillan14}
\bibinfo{author}{Kenneth~L. \surnamestart McMillan\surnameend}
  (\bibinfo{year}{2014}): \emph{\bibinfo{title}{Lazy Annotation Revisited}}.
\newblock In: {\sl \bibinfo{booktitle}{{CAV}}},
  \doi{10.1007/978-3-319-08867-9_16}.

\bibitemdeclare{inproceedings}{moura08}
\bibitem{moura08}
\bibinfo{author}{Leonardo~Mendon{\c{c}}a \surnamestart de~Moura\surnameend} \&
  \bibinfo{author}{Nikolaj \surnamestart Bj{\o}rner\surnameend}
  (\bibinfo{year}{2008}): \emph{\bibinfo{title}{{Z3:} An Efficient {SMT}
  Solver}}.
\newblock In: {\sl \bibinfo{booktitle}{{TACAS}}},
  \doi{10.1007/978-3-540-78800-3_24}.

\bibitemdeclare{inproceedings}{rummer13a}
\bibitem{rummer13a}
\bibinfo{author}{Philipp \surnamestart R{\"{u}}mmer\surnameend},
  \bibinfo{author}{Hossein \surnamestart Hojjat\surnameend} \&
  \bibinfo{author}{Viktor \surnamestart Kuncak\surnameend}
  (\bibinfo{year}{2013}): \emph{\bibinfo{title}{Classifying and Solving {Horn}
  Clauses for Verification}}.
\newblock In: {\sl \bibinfo{booktitle}{{VSTTE}}},
  \doi{10.1007/978-3-642-54108-7_1}.

\bibitemdeclare{inproceedings}{rummer13b}
\bibitem{rummer13b}
\bibinfo{author}{Philipp \surnamestart R{\"{u}}mmer\surnameend},
  \bibinfo{author}{Hossein \surnamestart Hojjat\surnameend} \&
  \bibinfo{author}{Viktor \surnamestart Kuncak\surnameend}
  (\bibinfo{year}{2013}): \emph{\bibinfo{title}{Disjunctive Interpolants for
  {Horn}-clause Verification}}.
\newblock In: {\sl \bibinfo{booktitle}{{CAV}}},
  \doi{10.1007/978-3-642-39799-8_24}.

\end{thebibliography}

\clearpage
\appendix

\section{Generating a Minimal CDD Expansion}
\label{app:cons-cdd}
\begin{algorithm}[t]
  \SetKwInOut{Input}{Input}
  \SetKwInOut{Output}{Output}
  \SetKwProg{myproc}{Procedure}{}{}
  \Input{A recursion-free CHC system $S$.}
  \Output{A minimal CDD expansion $S'$ of $S$ and
    a correspondence from $S'$ to $S$.} %
  \myproc{$\expand(S)$ %
    \label{line:expand-begin} }{ %
    \myproc{$\expandaux(S')$
      \label{line:expand-aux-begin} }{ %
      \Switch{$\sharingclause(S')$}{ %
        \lCase{$\none$:}{ \Return{$S'$} %
          \label{line:expand-ret} } %
        \lCase{$C \in S', %
          P \in \predof{S'}$: }{ %
          \Return{$\expandaux( %
            \copyrel(S', C, P))$ } %
          \label{line:expand-recurse}
        } %
      } %
    \label{line:expand-aux-end} } %
    \Return{$( \expandaux(S), \corr )$ %
      \label{line:expand-base-call}}
  } %
  \caption{$\expand$:
    given a recursion-free CHC system $S$, returns a minimal
    CDD expansion $S'$ of $S$ and its correspondence.}
  \label{alg:expand}
\end{algorithm}
%
Given a recursion-free CHC system $S$, \autoref{alg:expand} returns a
minimal CDD expansion of $S$ (\autoref{defn:min-expansion}).
$\expand$ defines a procedure $\expandaux$
(\autoref{line:expand-aux-begin}---\autoref{line:expand-aux-end}) that
takes a CHC system $S$ and returns a minimal CDD expansion of $S$.
$\expand$ runs $\expandaux$ on $S$ and
returns the result, paired with the map $\corr: \predof{S'} \to \predof{S}$ 
(\autoref{line:expand-base-call}).

$\expandaux$, given a recursion-free CHC system $S'$,
runs a procedure $\sharingclause$ on $S'$, which tries to
find a clause $C \in S'$ and a predicate $P \in \bodyof{C}$
such that $P$ is in the transitive dependencies of two sibling
predicates.
In such a case, we say that $(C, P)$ is a
\emph{sibling-shared dependency}.

If $\sharingclause$ determines that no sibling-shared dependency
exists, then $\expandaux$ returns $S'$
(\autoref{line:expand-ret}).

Otherwise, $\sharingclause$ must have located a sibling-shared
dependency $(C, P)$. In this case, $\expandaux$ runs $\copyrel $ on
$S'$, $C$, and $P$, which returns an expansion of $S'$ by creating a
fresh copy of $P$ and updating $\bodyof{C}$ to avoid the shared
dependency.
$\expandaux$ recurses on this expansion and returns the result
(\autoref{line:expand-recurse}).
%

$\expand$ always returns a CDD expansion of its input
(see \autoref{app:corr}, \autoref{lem:expand-corr}) that is minimal.
$\expand$ is certainly not unique as an algorithm for
generating a minimal CDD expansion.
In particular, feasible variations of $\expand$ can be
generated from different implementations of $\sharingclause$, each of
which chooses clause-relation pairs to return based on different
heuristics.
We expect that other expansion algorithms can also be developed by
generalizing algorithms introduced in previous work on generating
compact verification conditions of hierarchical
programs~\cite{lal-qadeer15}.

\section{Proof of characterization of CDD systems}
\label{app:char}
The following is a proof of \autoref{thm:cdd-contains}.
\begin{proof}
  To prove that CDD is a strict superset of the union of the class
  of linear systems and the class of body-disjoint systems, we prove
  \textbf{(1)} CDD contains the class of linear systems,
  \textbf{(2)} CDD contains the class of body-disjoint systems, and
  \textbf{(3)} there is some CDD system that is neither linear nor
  body-disjoint.

  For goal \textbf{(1)}, let $S$ be an arbitrary linear
  system.
  $S$ is CDD if for each clause $C$ in
  $S$ \textbf{(1)} and each pair of distinct predicates in
  the body of $C$ has disjoint transitive dependencies and
  \textbf{(2)} no predicate appears more than once in the body of $C$.
  (\autoref{defn:cdds}).
  Let $C$ be an arbitrary clause in $S$.
  Since $C$ is a linear clause, it has at most one relational
  predicate in its body. And since the system is recursion-free, the
  transitive dependencies are trivially disjoint and there can be no
  repeated predicate.
  Therefore, $S$ is CDD.

  For goal \textbf{(2)}, let $S$ be an arbitrary
  body-disjoint system.
  The dependence relation of $S$ is a tree $T$, by the
  definition of a body-disjoint system.
  Let $C$ be an arbitrary clause in $S$, with
  distinct relational predicates $R_0$ and $R_1$ in its body.
  All dependencies of $R_0$ and $R_1$ are in subtrees of $T$, which
  are disjoint by the definition of a tree.
  Thus, $S$ is CDD, by \autoref{defn:cdds}.

  For goal \textbf{(3)}, the system $\mcchc$ is CDD, but is neither linear
  nor body-disjoint.
\end{proof}

\section{Proof $\solvecdd$ is in co-NP}
\label{app:solve-cp}
The following is a proof of \autoref{thm:solve-cp}.
Namely, $\solvecdd$ is in co-NP.
\begin{proof}
  $\prectr$ and $\postctr$ construct formulas linear in the size of the
  CHC system. The satisfiability problem for the constructed formulas
  are in NP for linear arithmetic.
  $\solvecdd$ issues (at worst) a linear number of interpolation
  queries in terms of number of predicate.
  Therefore, the upper bound of $\solvecdd$ is co-NP.
\end{proof}

\section{Proof of Correctness}
\label{app:corr}
In this section, we prove that \sys is correct when applied to
recursion-free CHC systems.
We first establish lemmas for the correctness of each procedure used
by \sys, namely \textsc{Collapse} (\autoref{lem:expansion-sound} and
\autoref{lem:expansion-complete}), $\expand$
(\autoref{lem:expand-corr}), and $\solvecdd$
(\autoref{lem:vc}, \autoref{lem:cdd-soln-sound}, and
\autoref{lem:cdd-soln-complete}).
We combine the lemmas to prove \sys is correct (\autoref{thm:corr}).

For two recursion-free CHC systems $S$ and $S'$, if $S'$ is
an expansion of $S$, then the result of collapsing a
solution of $S'$ is a solution of $S$.
\begin{lem}
  \label{lem:expansion-sound}
  For two recursion-free CHC system $S'$ and $S$ such that $\sigma'$
  is a solution of $S'$ and $\eta$ is a correspondence from
  $S'$ to $S$, $\collapse{\eta}{\sigma'}$ is a
  solution of $S$.
\end{lem}
\begin{proof}
  %
  For each predicate $P' \in \predof{S'}$ such that $\eta(P') = P$,
  there must exist some clause $C' \in S'$ such that $P' \in \bodyof{C'}$
  because $S'$ is an expansion of $S$.
  Let predicate $Q' \in \predof{S'}$ be the head of $C'$.
  $\sigma'[P'] \land \consof{C'} \entails \sigma'[Q']$ by
  the fact that $\sigma'$ is a solution of $S'$.
  Therefore,
  \[ \collapse{ \eta }{ \sigma' }[P] \land \consof{C'} \entails %
  \left( \bigland_{ \substack{Q' \in \predof{S'} \\ \eta(Q') = Q} }
  \sigma(Q') = %
  \collapse{\eta}{\sigma'}[Q'] \right)
  \]
  Therefore, $\collapse{\eta}{\sigma'}$ has a solution for $P$.
  Since $\collapse{\eta}{\sigma'}$ has a solution for each predicate
  in $S$, $\collapse{\eta}{\sigma'}$ is a solution of $S$.
\end{proof}

Every expansion of a solvable recursion-free CHC system is also
solvable.
\begin{lem}
  \label{lem:expansion-complete}
  If a recursion-free CHC system $S$ is solvable and
  $S'$ is an expansion of $S$, then $S'$ is solvable.
\end{lem}
\begin{proof}
  Let $\sigma$ be a solution of $S$, and let $\eta$ be a
  correspondence from $S'$ to $S$.
  Let $\sigma'$ be such that for each $P' \in \predof{S'}$,
  $\sigma'(P') = \sigma(\eta(P'))$.
  Then $\sigma'$ is a solution of $S'$.
\end{proof}

$\expand$ always returns a CDD expansion of its input.
\begin{lem}
  \label{lem:expand-corr}
  For two recursion-free CHC systems $S$ and $S'$ and a correspondence
  from $S'$ to $S$, $\eta$, such that $(S', \eta) = \expand(S)$,
  $S'$ is a CDD system and an expansion of $S$.
\end{lem}
\begin{proof}
  By induction over the evaluation of $\expand$ on
  an arbitrary recursion-free CHC system $S$.
  The inductive fact is that for each evaluation step
  $\corr$ is a correspondence from argument
  $S'$ to $S$.
  In the base case, \expandaux is called initially on $S$,
  by \autoref{alg:expand}.
  $\corr$ is a correspondence from $S$ to itself, by the definition of
  $\corr$
  (\autoref{app:cons-cdd}).

  In the inductive case,
  $\expandaux$ constructs an argument
  $\copyrel(S, C, P)$
  where $C$ is a clause and $P$ is a predicate in $S$.
  $\expandaux$ recusively invokes itself with this argument.
  For each recursion-free CHC system $S'$ generated by $\copyrel(S, C, P)$,
  $\corr$ is a correspondence from $S'$ to
  $S$ by definition of $\copyrel$
  (\autoref{app:cons-cdd}).
  By this fact and the inductive hypothesis, $\corr$ is
  a correspondence from $\copyrel(S',
  C, P)$ to $S$.

  $\expandaux$ returns its parameter at some step, by
  \autoref{alg:expand}.
  Therefore, $\expandaux$ returns an expansion of $S$.

  For a given recursion-free CHC system $S'$, if $(S', \eta) =
  \expand(S')$,
  then $\sharingclause(S') = \none$, by the definition of $\expand$.
  If $\sharingclause(S') = \none$, then $S'$ is
  CDD, by the definition of $\sharingclause$ and CDD systems
  (\autoref{defn:cdds}).
  Therefore, $S'$ is CDD.
\end{proof}
Furthermore, $\expand$ returns a \emph{minimal} CDD
expansion of its input.
This fact is not required to prove \autoref{thm:corr}, and thus a
complete proof is withheld.

For each recursion-free CHC system $S$, $S$ has a solution if and only
if all interpolation queries return interpolants.
\begin{lem}
  \label{lem:vc}
  Given a recursion-free CHC system $S$ that is CDD and solvable, for
  all predicates $P \in \predof{S}$,
  $\solveitp$ returns
  an interpolant $I$.
\end{lem}
\begin{proof}
  Assume that $S$ has a solution $\sigma$ and there are some
  predicates $P \in \predof{S}$ such that
  $\solveitp$
  returns $SAT$. This means there must be a model $m$ for the
  conjunction of $\prectr(\cc{P},\sigma')$ and
  $\postctr(\cc{P},\sigma')$.
  But $\postctr(\cc{P},\sigma) = \false$, by the definition of
  a solution of a CHC system.
  Therefore, there can be no such model $m$.
  Therefore,
  $\solveitp$ always
  returns interpolant.
\end{proof}

\begin{lem}
  \label{lem:solve-aux}
  Given a recursion-free CHC system $S$ that is CDD, if for all
  predicates $P \in \predof{S}$,
  $\solveitp$ returns
  an interpolant $I$, then $\solvecdd(S)$ is a solution of $S$.
\end{lem}
\begin{proof}
  By induction on the $\solvecdd(S)$ calls to $\solveitp$ over all
  predicates $P \in \predof{S}$
  in topological order.
  The inductive fact is that after each call to $\solveitp$, $\sigma$
  is a partial solution of $S$.
  In the base case, $\predof{S} = \varnothing$,
  Therefore, $\sigma$ is a solution of $S$.

  In the inductive case, $\solvecdd$ calls $\solveitp$ on a predicate
  $P \in \predof{S}$ with partial solution $\sigma$.  Due to the
  topological ordering, $\sigma$ contains interpretations for each
  predicate $P \in \depsof{S}$.
  Based on the definition of an interpolant (\autoref{defn:itps}),
  $\prectr(\cc{P},\sigma)$ and $\postctr(\cc{P},\sigma)$ are
  inconsistent.
  The interpolant of these two formulas returned by $\solveitp$, $I$,
  is entailed by each clause $C$ where $P$ is the head where the
  predicates in $\bodyof{C}$ are substituted by their interpretations
  in $\sigma$.
  $I$ is also inconsistent with all constraints that appear after $P$
  that support the query clause.
  Therefore, when $\solvecdd$ updates $\sigma$ by binding $P$ to $I$
  the result is a partial solution of $S$.
\end{proof}

The output of $\solvecdd$ is correct for a given input CDD system.
\begin{lem}
  \label{lem:cdd-soln-sound}
  For a given CDD system $S$, and $\sigma = \solvecdd(S)$, %
  $\sigma$ is a solution of $S$.
\end{lem}
\begin{proof}
  The fact that $\solvecdd(S)$ returns $\sigma$ implies that for each
  predicate $P \in \predof{S}$,
  $\solveitp$ returns
  a valid interpolant (\autoref{alg:solve-cdd}).
  Therefore, \autoref{lem:solve-aux} implies that \solvecdd~returns a
  complete solution of $S$.
\end{proof}

\begin{lem}
  \label{lem:cdd-soln-complete}
  For a CDD system $S$ such that $S$ is
  solvable, %
  there is some $\sigma$ such that %
  $\sigma = \solvecdd(S)$.
\end{lem}
\begin{proof}
  For all predicates $P \in \predof{S}$, $\solveitp$ returns a
  interpolant $I$, by
  \autoref{lem:vc} and the fact that $S$ is solvable.
  Therefore, by \autoref{lem:solve-aux} and the fact that
  $S$ is solvable, $\solvecdd(S)$ returns a
  solution of $S$.
\end{proof}

The output of \sys is correct for a given input CDD system.
(\autoref{sec:approach}, \autoref{thm:corr}).
\begin{proof}
  Given two recursion-free CHC systems $S$ and $S'$ and a
  $\eta$ such that $(S', \eta) = \expand(S)$,
  $S'$ is minimal CDD expansion of $S$ and $\eta$ is a correspondence
  from $S'$ to $S$ (\autoref{lem:expand-corr}).
  Assume that $S$ is solvable. Then so is $S'$, by
  \autoref{lem:expansion-complete}.
  Therefore, there exists some $\sigma'$ such that $\sigma' =
  \solvecdd(S')$, by the definition of $\shara$.
  $\sigma'$ is a solution of $S'$, by \autoref{lem:cdd-soln-sound}.
  $\collapse{ \eta }{ \sigma' }$ is a solution of $S$, by
  \autoref{lem:expansion-sound}.
  Therefore, $\shara$ returns a valid solution of $S$.
\end{proof}

\end{document}